\RequirePackage{fix-cm}
\documentclass{svjour3}                     % onecolumn (standard format)
\smartqed  % flush right qed marks, e.g. at end of proof
\usepackage{graphicx}
\usepackage[round]{natbib}
\usepackage[utf8]{inputenc} 
\usepackage[T1]{fontenc}
\usepackage{lmodern}
\usepackage[fleqn]{mathtools}
\usepackage{rotating}
\usepackage{amssymb}
\usepackage{amsmath}
\usepackage{graphicx}
\usepackage{caption}
\usepackage{epstopdf}
\usepackage{stackrel}
\usepackage{soul}
\usepackage{color}

\newcommand*\diff{\mathop{}\!\mathrm{d}}
\renewcommand*{\figurename}{{\bf Fig.}}

\newtheorem{thm}{Theorem}
\newtheorem{lem}{Lemma}
\newtheorem{cor}{Corollary}
\journalname{Decision Theory}
\begin{document}

\title{The risk function of the goodness-of-fit tests for tail models.
}


\author{Ingo Hoffmann        \and
        Christoph J.\ B\"orner %etc.
}


\institute{Ingo Hoffmann \at
              %first address \\
              \email{Ingo.Hoffmann@hhu.de}           %  \\
			  %\emph{Present address:} of F. Author  %  if needed
           \and
           Christoph J.\ B\"orner \at
	          \email{Christoph.Boerner@hhu.de\\}
	          { }\\  
              Financial Services, Faculty of Business Administration and Economics, \\ 
              Heinrich Heine University D\"usseldorf, 40225 D\"usseldorf, Germany \\
              Tel.: +49 211 81-15258\\
              Fax:  +49 211 81-15316\\
              %\email{fauthor@example.com}           %  \\
			  %\emph{Present address:} of F. Author  %  if needed
}

\date{Version: \today { }}%/ Received: date / Accepted: date}

% Mögliche Konferenz: DAGStat der LMU in München (Deutsche Arbeitsgemeinschaft Statistik) 
% vom 18.-22.03.2019 in München

\maketitle
\thispagestyle{empty}
\begin{abstract}
This paper contributes to answering a question that is of crucial importance in risk management and extreme value theory: How to select the threshold above which one assumes that the tail of a distribution follows a generalized Pareto distribution. This question has gained increasing attention, particularly in finance institutions, as the recent regulative norms require the assessment of risk at high quantiles.
Recent methods answer this question by multiple uses of the standard goodness-of-fit tests. These tests are based on a particular choice of symmetric weighting of the mean square error between the empirical and the fitted tail distributions. Assuming an asymmetric weighting, which rates high quantiles more than small ones, we propose new goodness-of-fit tests and automated threshold selection procedures. We consider a parameterized family of asymmetric weight functions and calculate the corresponding mean square error as a loss function. We then explicitly determine the risk function as the finite sample expected value of the loss function. Finally, the risk function can be used to discuss the question of which symmetric or asymmetric weight function and, thus, which goodness-of-fit test should be used in a new method for determining the threshold value.
\end{abstract}

\section{Introduction}\label{intro}
In many disciplines, there is often a need to adapt a statistical model to the existing data to make statements about uncertain future outcomes. In particular, when assessing risks, an estimate of the major losses must be based on events that, although they have a low probability of occurrence, have a high impact. In the financial sector in particular, with its tightening regulatory requirements, models will be in demand that enable very good, qualitative and quantitative statements at high quantiles in the tail range of an underlying unknown distribution function.

Since the actual distribution of the data is often unknown, statisticians begin with a guess about the underlying statistical model for the entire value range or, more specifically, for the considered tail. They often use various distribution functions to choose the most suitable one later. In many cases, these models do not perfectly reflect the data. However, specific statistical tests can be applied to assess how good or bad a model fits the data, e.g., the Cram\'er-von Mises test \citep{cramer28, vonmises31} or Anderson-Darling test \citep{anderson52, anderson54}. 
These goodness-of-fit tests are also used in automated procedures to determine the threshold value at which the tail of the underlying distribution can be modeled using the generalized Pareto distribution \citep{bader18}.

All of these tests are based on the weighted mean square error $\hat R$, which corresponds to the loss function in decision theory \citep{aggarwal55, ferguson67}. Evaluated for a specific sample of length $n$, $\hat R_n$ calculates the weighted deviation of the modeled data from the measured data and, thus, the individual loss of accuracy by the model. At this level, the derived statistical tests above are used to assess the quality of the model. However, to be able to judge how good a statistical test is, the question to answer is how large the average loss is when all possible time series of measured data with length $n$ and unknown distribution functions are considered. An answer to this question is provided by the finite sample expectation value of the weighted mean square error ${\text E} [\hat R_n]$, which corresponds to the risk function in decision theory.

If the error is not squared but has an initially free exponent, \citet{aggarwal55} was able to explicitly calculate the risk function for this deviation error.
Different weight functions were considered for only two specific cases: the Cram\'er-von Mises and Anderson-Darling tests.
In particular, evaluating models for the upper {\it or} lower tail, weight functions are important, which enables only a stronger weighting of deviations in these areas of the distribution.
%Editor's note: Please ensure that the intended meaning has been maintained in the above edit. Angepasst 29.06.2018 IH.
These weighting functions define the families of special tail statistics on which we focus here. The question remains of which statistics of this parameterized family should be used for the present task to establish a suitable goodness-of-fit test in an automated method for determining the threshold value.

As the main result of this analysis, we calculate the risk function for this family of tail statistics, which allows us to compare different statistics in terms of their average loss. Thus, the question of a suitable statistic for a tail-oriented goodness-of-fit test can be discussed. Our result shows that some statistics diverge and cannot be used. The results further suggest that from theoretical and practical viewpoints, the statistics first suggested by \citet{ahmad88} should be chosen as a goodness-of-fit test for analyzing the tail and evaluating a tail model. This statistic should be further investigated and used as the origin of an automated method for determining the threshold to separate the tail from the distribution.

The remainder of the paper is structured as follows: After defining the family of tail statistics in Section \ref{tailstatistics}, the corresponding risk function is explicitly calculated in Section \ref{riskfunction}. Section \ref{corollaries} summarizes some corollaries that follow from the theorem of the previous section. As an interesting side result, we define a one-parametric discrete distribution function over a finite support of non-negative integers and determine all moments of this distribution, which may be useful in a decision-theoretical problem where probabilities are to be assigned to a limited number of environmental states. The final section discusses the results and summarizes the key points.

\section{Definition of tail statistics} \label{tailstatistics} 
Let $X_1, X_2, \ldots , X_n$ be a sample of random variables with a common unknown continuous distribution function $F(x)$ and density function $f(x)$. The corresponding empirical distribution function for $n$ observations is defined as
\begin{flalign}\label{EDFEQ}
F_n(x)& = \frac 1 n \sum_{i = 1}^n {\boldsymbol 1} (X_i \leq x),
\end{flalign}
where ${\boldsymbol 1}$ is the indicator function; ${\boldsymbol 1} (X_i \leq x)$ is equal to one if $X_i\leq x$ and zero otherwise. Thus, $F_n(x) = \frac k n$ if $k$ observations are less than or equal to $x$ for $k=0,1, \ldots ,n$ \citep{kolmogorov33}.\\

As a convenient measure of the discrepancy or ''distance'' between the distribution functions $F_n(x)$ and $F(x)$, we consider the weighted mean square error
\begin{flalign}\label{WMSEEQ}
\hat R_n & = n \int_{-\infty}^{+\infty} \left(F_n(x)-F(x) \right)^2\; w(F(x))\;\diff F(x),
\end{flalign}
introduced in the context of statistical test procedures by \citet{cramer28}, \citet{vonmises31} and \citet{smirnov36}. 
The non-negative weight function $w(t)$ in Eq.\ (\ref{WMSEEQ}) is a suitable preassigned function for accentuating the difference between the distribution functions in the range where the test procedure is desired to have sensitivity. Consider the weight function
\begin{flalign}\label{WFEQ}
w(t) & = \frac{1}{t^a(1-t)^b}
\end{flalign}
for free real-valued stress parameters $a,b\geq 0$ and $t\in[0,1]$. Here, $a$ affects the weight at the lower tail, and $b$ affects the weight at the upper tail. These stress parameters, at a certain position on the distribution function, allow one to change the strength with which the deviations from the empirical distribution function at that position are weighted. Put simply, by using the stress parameters, the magnification is adjusted, with which the deviation between the distributions at a fixed position is considered.

Then, for $a=b=0$, Eq.\ (\ref{WMSEEQ}) provides the Cram\'er-von Mises statistic \citep{cramer28, vonmises31}, and when both tails are heavily weighted ($a=b=1$), it is equal to the Anderson-Darling statistic \citep{anderson52, anderson54}. The Anderson-Darling statistic simultaneously weights the difference between the  distributions more heavily at both ends of the  distribution $F(x)$.

Mixed weight functions can hinder the individual study of one tail or the other of the distribution function. In particular, in the construction of goodness-of-fit tests that focus on a tail, pure functions, which weight one side of the distribution function strongly, are beneficial. As the regulatory requirements become more stringent, statistics may become increasingly interesting, which weight the differences in either the upper or lower tail of the distribution function more strongly.
Therefore, the following weight functions should gain importance.\\
\\
The weight function for the lower tail ($a\geq 0, b=0$) is
\begin{flalign}\label{WFLTEQ}
w(t) & = \frac{1}{t^a} 
\end{flalign}
\\
The weight function for the upper tail ($a=0, b\geq 0$) is
\begin{flalign}\label{WFUTEQ}
w(t) & = \frac{1}{(1-t)^b} 
\end{flalign}

If we initially leave the stress parameters indeterminate in the calculation of the weighted mean square error Eq.\ (\ref{WMSEEQ}), two families of statistics can be derived: one family is for the lower tail, and the other is for the upper tail.
For $a=b$, these two families can be transformed into each other using coordinate transformation $Z=-X$ of the random variable.
Therefore, in the following, we only treat the statistics family for stress parameter $a$. The derived results then apply to the second family for parameter $b$.\\

With Eq.\ (\ref{WFLTEQ}), the weighted mean square error Eq.\ (\ref{WMSEEQ}) reduces to
\begin{flalign}\label{LTWMSEEQ}
\hat R_{n,a} & = n \int_{-\infty}^{+\infty} \frac{\left(F_n(x)-F(x) \right)^2}{\left(F(x)\right)^a}\; \diff F(x).
\end{flalign}
The computing formulae for this family of lower-tail statistics can be obtained by following the method given in \citet{anderson54}.

Let $x_{(1)} \leq x_{(2)}\leq \ldots \leq x_{(n)}$ be the sample values (in ascending order) obtained by ordering each realization  $x_1, x_2, \ldots , x_n$ of $X_1, X_2, \ldots , X_n$. Then, we can summarize the following calculation rules for the statistics:\\
\\
$\bullet$  {$a\neq 1,2,3$}%
\begin{flalign}\label{LTSEEQ}
\hat R_{n,a} & = \frac{2}{(1-a)(2-a)(3-a)}\; n \\\nonumber
& + \sum_{i=1}^n\left[ \frac{2}{2-a}\left(F(x_{(i)})\right)^{2-a} - \frac{2i-1}{n}\frac{1}{1-a} \left(F(x_{(i)})\right)^{1-a} \right]
\end{flalign}
%%
%$\bullet$ {$a= 0$}\\
%\\
Note: In the special case where $a=0$, Eq.\ (\ref{LTSEEQ}) reduces to the statistics $W_n^2\; (= \hat R_{n,0})$ proposed by \citet{cramer28} and \citet{vonmises31}:
\begin{flalign}\label{CVM}
W_n^2  & = \frac{1}{12n}  
+ \sum_{i=1}^n\left[ \frac{2i-1}{2n} - F(x_{(i)})  \right]^2
\end{flalign}
$\bullet$ {$a= 1$}%
\begin{flalign}\label{LTSEEQa1}
\hat R_{n,1} & = -\frac{3}{2} n 
+ \sum_{i=1}^n\left[ 2 F(x_{(i)}) - \frac{2i-1}{n} \ln\left(F(x_{(i)})\right) \right]
\end{flalign}
To obtain an appropriate goodness-of-fit test specifically for the tail of a distribution, the computation formulae Eq.\ (\ref{LTSEEQa1}) were first described by \citet{ahmad88} and later examined more formally by the same authors with regard to the distribution of their test statistics $AL_n^2$ \citep*{sinclair90}.\\
\\
$\bullet$ {$a=2$}%
\begin{flalign}\label{LTSEEQa2}
\hat R_{n,2} & = 
\sum_{i=1}^n\left[\frac{2i-1}{n}\frac{1}{F(x_{(i)})}
+ 2 \ln\left(F(x_{(i)})\right) \right]
\end{flalign}
$\bullet$ {$a=3$}\\
\\
For the stress parameter $a = 3$, no feasible solution can be calculated because $\hat R_{n,3}$ approaches infinity.
\section{Risk function} \label{riskfunction}
In decision theory, the weighted mean square error $\hat R_n$, which is defined in the previous section (see Eq.\ (\ref{WMSEEQ})), is generally referred to as the loss function, and the expected value of the loss function is called the risk function \citep{aggarwal55, ferguson67}:
\begin{flalign}\label{RFRn}
R_n & = {\text E} \left[ \hat R_n  \right]
\end{flalign}
For the case considered here, the risk function can be calculated explicitly. Our main result summarizes the following theorem and the complementary corollaries in section \ref{corollaries}.
\begin{thm}\label{lemmariskfunction}
	Let $\hat R_{n,a}$ be the weighted mean square error defined by Eq.\ (\ref{LTWMSEEQ}). Then, $\forall a\in \mathbb{R}^{\geq 0}$, the risk function is given by
	\begin{flalign}\label{RFLT}
	R_{n,a} = \frac{1}{(2-a)(3-a)}.
	\end{flalign}
\end{thm}
\begin{proof}
Using the transformation $u=F(x)$, the lower tail statistics can be expressed in terms of $u \in [0, 1]$, and $u_{(1)} \leq u_{(2)}\leq \ldots \leq u_{(n)}$ is an ordered sample of size $n$ from a continuous uniform distribution over the interval $[0, 1]$. The expectation in Eq.\ (\ref{RFRn}) must be taken with respect to this distribution. Since the distribution of the $i$th-order statistic $U_{(i)}$ in a random sample of size $n$ from the uniform distribution over the interval $[0, 1]$ is a beta distribution with the following probability density 
\begin{flalign}\label{OSPD}
p(u) = \frac{1}{B(i, n-i+1)} u^{i-1}(1-u)^{n-i}
\end{flalign}
the expectation value for $\hat R_{n,a}$ can be calculated as follows:\\
\\
$\bullet$ {$a\neq 1,2,3$}\\
\begin{flalign}\label{lem2Profa}
R_{n,a} & = {\text E} \left[ \hat R_{n,a}\right] \\ \nonumber
& = \frac{2n}{(1-a)(2-a)(3-a)}\;  \\ \nonumber
& \qquad + \sum_{i=1}^n \frac{2}{2-a}{\text E} \left[u_{(i)}^{2-a}\right] \\ \nonumber
& \qquad - \sum_{i=1}^n \frac{2i-1}{n}\frac{1}{1-a} {\text E} \left[u_{(i)}^{1-a}\right] \\ \nonumber
& = \frac{2n}{(1-a)(2-a)(3-a)}\;  \\ \nonumber
& \qquad + \sum_{i=1}^n \frac{2}{2-a}\frac{\int_0^1 u^{i+1-a}(1-u)^{n-i} \diff u }{B(i, n-i+1)} \\ \nonumber
& \qquad - \sum_{i=1}^n \frac{2i-1}{n}\frac{1}{1-a} \frac{\int_0^1 u^{i-a}(1-u)^{n-i} \diff u }{B(i, n-i+1)} \\ \nonumber
& = \frac{2n}{(1-a)(2-a)(3-a)}\;  \\ \nonumber
& \qquad + \sum_{i=1}^n \frac{2}{2-a}\frac{B(i+2-a, n-i+1) }{B(i, n-i+1)} \\ \nonumber
& \qquad - \sum_{i=1}^n \frac{2i-1}{n}\frac{1}{1-a} \frac{B(i+1-a, n-i+1)}{B(i, n-i+1)}  \nonumber
\end{flalign}
To evaluate the remaining sums, we use the following lemma.
\begin{lem}\label{identlem1}
\begin{flalign}\label{ident1}
m_k(\nu) & \stackrel{\text{\rm def}}{=} \frac{\nu+1}{n}\sum_{i=1}^n i^k \;\frac{B(i+\nu, n-i+1) }{B(i, n-i+1)} \\ \nonumber
		 & = \sum_{l=0}^{k}S_{k+1,l+1}\; \frac{\nu+1}{\nu+1+l} \;(n-1)_{(l)} ,
\end{flalign}
where $\nu\in \mathbb{R}$, $k\in\mathbb{N}$, $(n-1)_{(l)}$ is the Pochhammer notation for falling factorials and $S_{k,l}$ are the Stirling numbers of the second kind {\rm \citep{abramowitz14}}.
\end{lem}
\begin{proof}[Lemma \ref{identlem1}]
To begin, the beta functions $B(\cdot,\cdot)$ of Eq.\ (\ref{ident1}) are expressed in terms of the gamma function $\Gamma(\cdot)$ \citep{abramowitz14}. Simplifying the fraction yields
\begin{flalign}\label{ident1proof1}
	m_k(\nu) & = 	(\nu+1)\; \sum_{i=1}^n i^k \;\frac{\Gamma(n)}{\Gamma(i)} \frac{\Gamma(i-1+\nu+1)}{\Gamma(\phantom{-{}}n \phantom{1.}+\nu +1)}.
\end{flalign}
Depending on $\nu$, the possibly resulting poles must be considered, and the gamma function should be considered in its analytic continuation $ \Gamma(x+\alpha)= (x)^{(\alpha)}\Gamma(x)$, where $(x)^{(\alpha)}$ is the Pochhammer notation for rising factorials \citep{abramowitz14}. After simplifying the fraction, it follows that
\begin{flalign}\label{ident1proof2}
m_k(\nu) = (\nu+1)\; \sum_{i=1}^n i^k\;\frac{\Gamma(n)}{\Gamma(i)} \frac{1}{(\nu+1+[i-1])^{(n-[i-1])}}.
\end{flalign}
Using the identity $(x)^{(\alpha)} = (x)^{(\beta)} (x+\beta)^{(\alpha- \beta)} $ results in
\begin{flalign}\label{ident1proof3}
m_k(\nu) = (\nu+1)\; \sum_{i=1}^n i^k\; \frac{\Gamma(n)}{\Gamma(i)} \frac{(\nu+1)^{(i-1)\phantom{n}}}{(\nu+1)^{(n)\phantom{i-1}}}.
\end{flalign}	
Remember that $\frac{\Gamma(n)}{\Gamma(i)} = \binom{n}{i} \frac{i}{n} (n-i)! = \binom{n-1}{i-1} (1)^{(n-i)}$. Then,
\begin{flalign}\label{ident1proof4}
m_k(\nu) = \frac{(\nu+1)^{\phantom{(n)}}}{(\nu+1)^{(n)}}\; \sum_{i=1}^n i^k\;  \binom{n-1}{i-1} (1)^{(n-i)} (\nu+1)^{(i-1)}.
\end{flalign}
Now, $i^k$ is decomposed into a sum of the falling factorials, where the coefficients consist of Stirling numbers of the second kind $ i^k = \sum_{l=0}^{k} S_{k,l} (i)_{(l)}$. Using the appropriate numbering of the sum with an appropriate extension of the terms gives
\begin{flalign}\label{ident1proof5}
m_k(\nu)  &=  \frac{(\nu+1)^{\phantom{(n)}}}{(\nu+1)^{(n)}} \times \\ \nonumber \; &\sum_{i=1}^n \sum_{l=1}^{k+1} S_{k+1,l} (n-1)_{(l-1)} \frac{(i\phantom{n}\mspace{-7mu}-1)_{(l-1)}}{(n\phantom{i}\mspace{-7mu}-1)_{(l-1)}}  \binom{n-1}{i-1} (1)^{(n-i)} (\nu+1)^{(i-1)}. 
\end{flalign}
By changing the order of the sums and truncating the binomial coefficient, the above equation reduces to
\begin{flalign}\label{ident1proof6}
m_k(\nu)  &=  \frac{(\nu+1)^{\phantom{(n)}}}{(\nu+1)^{(n)}} \times \\ \nonumber \; &\sum_{l=1}^{k+1} S_{k+1,l} (n-1)_{(l-1)} \sum_{i=1}^n \binom{n-l}{i-l} (1)^{(n-i)} (\nu+1)^{(i-1)}. 
\end{flalign}
By renumbering the last sum, splitting the term $ (\nu+1)^{(l-1)}$ and using Chu-Vandermonde theorem \citep[Ch.\ 18]{oldham09}, this equation reduces to
\begin{flalign}\label{ident1proof7}
m_k(\nu)  &=  \frac{(\nu+1)^{\phantom{(n)}}}{(\nu+1)^{(n)}} \times \\ \nonumber \; &\sum_{l=1}^{k+1} S_{k+1,l} (n-1)_{(l-1)} (\nu+1)^{(l-1)} (\nu+1+l)^{(n-l)}. 
\end{flalign}
After renumbering of the remaining sum and multiplying the raising factorials, Eq.\ (\ref{ident1}) follows immediately. \hfill {\it (Lemma \ref{identlem1})}\qed
\end{proof}%Ende Proof Lemma 1 
With Lemma \ref{identlem1}, we continue the proof of Theorem \ref{lemmariskfunction}. 
Let 
\begin{flalign}\label{lem2ProfDef}
H_k(\nu) & \stackrel{\text{\rm def}}{=} \frac{n}{\nu+1} m_k(\nu).
\end{flalign}
Now, Eq.\ (\ref{lem2Profa}) becomes
\begin{flalign}\label{lem2Prof2}
R_{n,a} & = \frac{2n}{(1-a)(2-a)(3-a)}\;  \\ \nonumber
& \qquad +  \phantom{\frac{1}{n}}\frac{2}{2-a} H_0(2-a)  \\ \nonumber
& \qquad +  \frac{1}{n} \frac{1}{1-a}          H_0(1-a)  \\ \nonumber
& \qquad  - \frac{2}{n} \frac{1}{1-a}          H_1(1-a) .
\end{flalign}
After a few algebraic transformations, this equation reduces to Eq.\ (\ref{RFLT}).\\
\\
Finally, let us examine the special cases:\\
\\
$\bullet$ {$a=1$}\\
\begin{flalign}\label{lem2Profa1}
R_{n,1} & = {\text E} \left[ \hat R_{n,1}\right] \\ \nonumber
& = -\frac{3n}{2} + \sum_{i=1}^n 2\;{\text E} \left[u_{(i)}\right] 
- \sum_{i=1}^n \frac{2i-1}{n}\; {\text E} \left[\ln{u_{(i)}}\right] \\ \nonumber
& = -\frac{3n}{2} + \sum_{i=1}^n 2\frac{B(i+1, n-i+1)}{B(i, n-i+1)} 
- \sum_{i=1}^n \frac{2i-1}{n}  \Big( \psi(i) - \psi(n+1) \Big) \\ \nonumber
& = -\frac{3n}{2} + 2 H_0(1) - \sum_{i=1}^n \frac{2i-1}{n}  \Big( \psi(i) - \psi(n+1) \Big)
\end{flalign}
For the last sum, we use the computation formulae presented by 
\citet[Eq.\ (59) therein]{aggarwal55}, with $\psi(i)$ being the digamma function, cf.\ \citet{abramowitz14}. Then,
\begin{flalign}\label{lem2Profa12}
R_{n,1} & = -\frac{3n}{2} + n   +\frac{n}{2} +\frac{1}{2} \\ \nonumber
& = \frac{1}{2}
\end{flalign}
This result is also yielded by Eq.\ (\ref{RFLT}) for $a = 1$.\\
\\
\\
\\
\\
\\
\\
\\
%\newpage
$\bullet$ {$a=2$}\\
\begin{flalign}\label{lem2Profa2} %Die Funktion m ist hier neu zu betrachten
R_{n,2} & = {\text E} \left[ \hat R_{n,2}\right] \\ \nonumber
& = \sum_{i=1}^n \frac{2i-1}{n}{\text E} \left[\frac{1}{u_{(i)}}\right] 
+ \sum_{i=1}^n 2\; {\text E} \left[\ln{u_{(i)}}\right] \\ \nonumber
& = \sum_{i=1}^n \frac{2i-1}{n} \frac{B(i-1, n-i+1)}{B(i, n-i+1)} 
+ \sum_{i=1}^n 2\; \Big( \psi(i) - \psi(n+1) \Big) \\ \nonumber
& = \frac{1}{n} \Big(2 H_1(-1) - H_0(-1)\Big) -2n \\ \nonumber
\end{flalign}
The last expression shows that the risk function becomes infinite because $H_0(\nu)$ and $H_1(\nu)$ have a pole for $\nu = -1$ that remains even after the subtraction. This result is also described by the pole in Eq.\ (\ref{RFLT}) when $a = 2$.\\
\\
$\bullet$ {$a=3$}\\
\\
In Section \ref{tailstatistics}, we show that $\hat R_{n,3}$ approaches infinity. This also applies to the expected value $R_{n,3}$. This result is described by the pole in Eq.\ (\ref{RFLT}) when $a = 3$.
	{  } 	\hfill {\it (Theorem \ref{lemmariskfunction})} \qed
\end{proof} %Ende Proof Theorem

\section{Corollaries}\label{corollaries}
Because of the symmetry of the two families of statistics from Section \ref{tailstatistics}, we note the following:
\begin{cor}\label{RFUTCorrolar}
	Let $\hat R_{n,b}$ be the weighted mean square error defined by Eq.\ (\ref{WMSEEQ}) with the weight function in Eq.\ (\ref{WFUTEQ}). Then, $\forall b\in \mathbb{R}^{\geq 0}$, the risk function is
	\begin{flalign}\label{RFUT}
	R_{n,b} = \frac{1}{(2-b)(3-b)}.
	\end{flalign}
\end{cor}
\begin{proof}
	By direct calculation with regard to the coordinate transformation $Z=-X$ of the random variable. \qed
\end{proof}
In the special cases of the Cram\'er-von Mises statistic and the Anderson-Darling statistic, the following two corollaries hold:
\begin{cor}\label{CVMCorrolar}
	Let $W_n^2=\hat R_{n,\mathrm{CM}}$ (Eq.\ (\ref{WMSEEQ}) with weight functions Eq.\ (\ref{WFLTEQ}) for $a=0$) be the Cram\'er-von Mises statistic Eq.\ (\ref{CVM}). Then, the risk function is 
	\begin{flalign}\label{CVMRF}
	R_{n,\mathrm{CM}} = {\text E} \left[ W_n^2  \right]= \frac{1}{6}.
	\end{flalign}
\end{cor}
\begin{proof}
	We have $\hat R_{n,\mathrm{CM}} = \hat R_{n,a=0}$. Hence,
	\begin{flalign}\label{CMRFproof}
	R_{n,\mathrm{CM}} = {\text E} \left[ W_n^2  \right] = {\text E} \left[ \hat R_{n,\mathrm{CM}}  \right] = {\text E} \left[ \hat R_{n,a=0}   \right] \stackrel{\text{Eq.\ (\ref{RFLT})}}{=} \frac{1}{6}.
	\end{flalign}
	Note: The risk function calculated here is in accordance with the result of \citet[p.\ 453 below]{aggarwal55}. Note that \citet{aggarwal55} defined the loss function without multiplication by $n$.\qed
\end{proof}
\begin{cor}\label{ADCorrolar}
	Let $A_n^2=\hat R_{n,\mathrm{AD}}$ be the Anderson-Darling statistic (Eq.\ (\ref{WMSEEQ}) with the weight function in Eq.\ (\ref{WFEQ}) for $a=b=1$). Then, the risk function is 
	\begin{flalign}\label{ADRF}
	R_{n,\mathrm{AD}} = {\text E} \left[ A_n^2  \right] = 1.
	\end{flalign}
\end{cor}
\begin{proof}
	Because of the identity $w(t) = \frac{1}{t(1-t)}= \frac{1}{t} + \frac{1}{1-t}$, the expression $\hat R_{n,\mathrm{AD}} = \hat R_{n,a=1} + \hat R_{n,b=1} $ is equal to the well-known Anderson-Darling statistic $A_n^2$ \citep{anderson52, anderson54}. Hence,
	\begin{flalign}\label{ADRFproof}
	R_{n,\mathrm{AD}} = {\text E} \left[ A_n^2  \right] = {\text E} \left[ \hat R_{n,\mathrm{AD}}  \right] ={\text E} \left[ \hat R_{n,a=1}\right] + {\text E} \left[ \hat R_{n,b=1}   \right] \stackrel{\text{Eqn.\ (\ref{RFLT}, \ref{RFUT})}}{=} 1.
	\end{flalign}
	Note: The risk function calculated here is in accordance with the result of \citet[p.\ 461 above]{aggarwal55}.\qed
\end{proof}
For further considerations in the context of decision theory, the following corollary and the discrete probability distribution defined in it may be of interest because they can be used as a parametric probability model in a decision problem for a finite number of environmental states.
\begin{cor}\label{Stochastic} 
	Let $X$ be a discrete random variable with finite range $i=1,\ldots, n$. Then,	
	\begin{flalign}\label{ProbVector}
	p_i(\nu) & \stackrel{\text{\rm def}}{=} \frac{\nu+1}{n} \frac{B(i+\nu, n-i+1) }{B(i, n-i+1)}, 
	\end{flalign}
	where $\nu\in \mathbb{R}^{>-1}$, and $n \in\mathbb{N}^{>0}$ assigns a probability to each value in the range of $X$. The tuple $\vec p(\nu) \in [0, 1]^{n} $ with entries $p_i(\nu) \in [0, 1]$ is a probability vector that, depending on $\nu$, forms a family of discrete probability distributions on the finite support of non-negative integers with a cumulative distribution function
	\begin{flalign}\label{CDF}
	F(s;\nu)  & = \frac{B(n,\nu+1)}{B(s,\nu+1)}
	\end{flalign}
	where $s=1,\ldots, n$.\\
	\\
	For the discrete random variable, the following applies:
	\begin{enumerate}
		\item $k$-Moment
		\begin{flalign}\label{MomX}
		m_k(\nu)  \phantom{\text{\rm Var E}}\hspace{-2.0em}= 
		\sum_{l=0}^{k}S_{k+1,l+1}\; \frac{\nu+1}{\nu+1+l} \;(n-1)_{(l)}
		\end{flalign}
		\item Expectation
		\begin{flalign}\label{ExpX}
		{\text{\rm E}}\left[ X\right] \phantom{\text{\rm Var}m_k}\hspace{-2.0em} = 
		1 + \frac{\nu+1}{\nu+2}(n-1)
		\end{flalign}
		\item Variance
		\begin{flalign}\label{VarX}
		{\text{\rm Var}} \left[ X\right] \phantom{\text{\rm E}m_k}\hspace{-2.0em} = \frac{(\nu+1)(\nu+n+1)(n-1)}{(\nu+2)^2(\nu+3)}
		\end{flalign}
		
	\end{enumerate}
\end{cor}
\begin{proof} 
Eq.\ (\ref{ident1}) of Lemma \ref{identlem1} corresponds to the $k$-moments of the discrete probability distribution in Eq.\ (\ref{MomX}) and proves for $k=0$ that $\sum_i p_i(\nu) =1$. 
The alternative form
\begin{flalign}\label{DistrAltern}
	 p_i(\nu) = \frac{\Gamma(n)}{\Gamma(i)} \frac{\phantom{(}\nu+1\phantom{)^{(n-i+1)}}}{(\nu+i)^{(n-i+1)}}
\end{flalign}
(cf.\ Eq. (\ref{ident1proof2}) with $k=0$) shows that $p_i(\nu) \geq 0$ for $\nu >-1$.	 
Eq.\ (\ref{ident1}) for $k=1$ gives the first moment of the discrete probability distribution and is equal to the expectation value in Eq.\ (\ref{ExpX}). 
With ${\text{\rm Var}} \left[ X\right] = {\text{\rm E}} \left[ X^2\right] - ({\text{\rm E}} \left[ X\right])^2$, only the calculation of the second moment ${\text{\rm E}} \left[ X^2\right]$ remains. This calculation can be performed quickly using equation (\ref{ident1}). After summarizing the terms, Eq.\ (\ref{VarX}) follows. \\
The cumulative distribution function is obtained when all $p_i(\nu)$ are summed for $i=1,\ldots,s$; where $s \leq n$. Beginning with the alternative form of $p_i(\nu)$ (cf.\ Eq.\ (\ref{ident1proof1})) and using the techniques described in the proof of Lemma \ref{identlem1}, the following holds
\begin{flalign}\label{CDFproof}
F(s;\nu)  & = \Gamma(n)\frac{(\nu+1)}{\Gamma(n +\nu +1)}\;\; \sum_{i=1}^s \frac{\Gamma(i+\nu)}{\Gamma(i)}. \\ \nonumber
& = \frac{\Gamma(n)}{\Gamma(s)} \frac{(\nu +1)}{\Gamma(n +\nu +1)} \; \sum_{i=1}^s \frac{\Gamma(s)}{\Gamma(i)} (\nu)^{(i)} \Gamma(\nu) \\ \nonumber
& = \frac{\Gamma(n)}{\Gamma(s)} \frac{(\nu +1) \Gamma(\nu)}{\Gamma(n +\nu +1)} \; \sum_{i=1}^s \binom{s-1}{i-1} (1)^{(s-i)} (\nu)^{(i)} \\ \nonumber
& = \frac{\Gamma(n)}{\Gamma(s)} \frac{(\nu+2)^{(s-1)}(\nu +1) \nu \Gamma(\nu)}{\Gamma(n +\nu +1)} \\ \nonumber
& =\frac{\Gamma(n)}{\Gamma(s)} \frac{\Gamma(s+\nu+1)}{\Gamma(n +\nu +1)}
\end{flalign} 
The last expression is an alternative representation of Eq.\ (\ref{CDF}).\qed
\end{proof}

The above defined discrete distribution can weight either the left edge, with $ \nu\in]-1, 0[$, or the right edge, with $ \nu\in]0, +\infty[$, of the interval $i\in[1,n]$ more strongly. For $\nu = 0$, the discrete uniform distribution is reproduced. 

The general distribution Eq.\ (\ref{ProbVector}) can be used in models that occur in the context of decisions under risk. The limited number of environmental states is sorted in ascending or descending order according to their probability of occurrence. If the discrete probability distribution is to have certain properties, for example, for a given expectation value between 1 and $n$, this distribution can be modeled with the parameter $\nu$. 

Note that the discrete probability distribution is similar but not equal to the beta-binomial distribution $BeB(n,\alpha, \beta)$ for definition cf.\ \citep{abramowitz14}. The pivotal difference is that the origin of the above distribution is based on the distribution of order statistics and not on the binomial distribution. Only for $\alpha = \nu +1 +\frac{1}{n}$ and $\beta = 1 -\frac{1}{n} $ at very large $n$ do both distributions have nearly the same properties.
\newpage
\section{Discussion and Conclusion}
To summarize the results, Fig.\ \ref{riskfunctionbild} shows the dependence of the risk function on the stress parameter.
\begin{figure}[htbp]
	%\centering
	\captionsetup{labelfont = bf, labelsep = none}
	\includegraphics[width=0.975\textwidth]{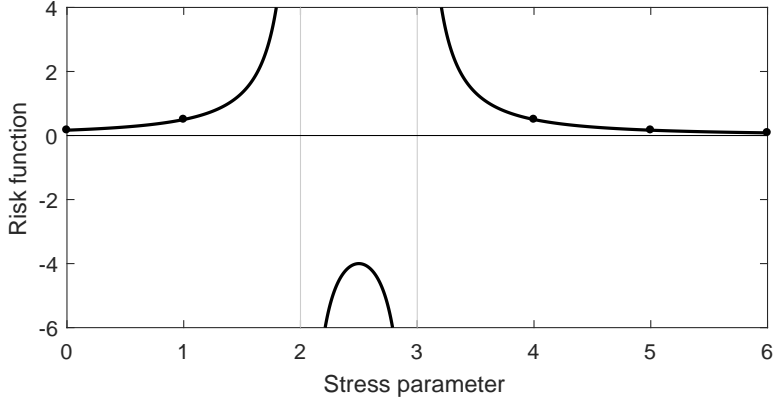}
	%\begin{quote}
		\caption[Sample-Size 100]{\label{riskfunctionbild} The risk function depending on the stress parameter (for $a$ or $b$). For integer values, the corresponding risk is marked with bullets. The two poles are marked with thin lines.}
	%\end{quote}
\end{figure}
The risk function is symmetric about the local maximum at $a=2.5$, which also applies for the second family with parameter $b$, and it has two poles where the sign changes.
If deviations in the tail region of a distribution function are to be weighted more heavily, stress parameters that are greater than zero are a suitable choice.
%Editor's note: Please ensure that the intended meaning has been maintained in the above edit. Rückgeändert IH29.06.2018.
Focusing only on integer values for the stress parameter, the result for $a = 2$ is surprising. Since the risk function approaches infinity for these values, the associated weighting function and corresponding statistic should not be used. In fact, during our research, we occasionally found excerpts of anonymous scripts that propose these statistics.
Because of the symmetry, the stress parameters for $a = 1$ and $a = 4$ are equivalent with respect to the risk function. Only for $a\geq 5$ can marginal improvements be achieved. 
However, in the preliminary investigations, for large exponents and small samples of financial data, the evaluation of the corresponding statistics became numerically difficult.
Therefore, we suggest using statistics with $a = 1$ proposed by \citet{ahmad88} as the basis for a goodness-of-fit test, particularly for the tail of a distribution. By doing so, similar to \citet{bader18}, a new automated method for determining the threshold value can be defined. Above the threshold, the tail of the distribution can then be modeled using the generalized Pareto distribution to calculate the required high quantiles. 

If the proposed statistic is used for $a = 1$, the average loss $R_{n, 1} = \frac{1}{2}$ is slightly larger than that for the Cram\'er-von Mises statistic ($R_{n,\mathrm{CM}} = \frac{1}{6}$) but is much smaller than that for the Anderson-Darling statistic ($R_{n,\mathrm{AD}} = 1$).
The average losses detected by the risk function can be further minimized if the empirical distribution function Eq.\ (\ref{EDFEQ}) is not used to determine the weighted mean square error. Instead, to improve the results, the empirical probability may be evaluated depending on the selected weight function, cf. \ e.g., \ \citet{ferguson67}, which again leads to new families of statistics and risk functions and opens another field of research.

%\begin{acknowledgements}
%\end{acknowledgements}

\bibliographystyle{plainnat}


\end{document}